\documentclass{article}

\usepackage{graphicx}
\usepackage{microtype}
\usepackage{caption}
\usepackage[latin1]{inputenc}
\usepackage{amssymb,amsmath,array}
\usepackage{amsthm}
\usepackage{epsfig}
\usepackage{url}
\usepackage{subcaption}
\usepackage[algo2e,linesnumbered,boxed]{algorithm2e} 
\usepackage[hidelinks]{hyperref}
\newcommand{\MR}[1]{\mathrm{#1}}
\newcommand{\MB}[1]{\mathbf{#1}}

\newcommand{\MC}[1]{\mathcal{#1}}

\newcommand{\xj}[1]{$\vec{x}_{#1}$}

\newcommand{\Tr}[0]{\mathrm{Tr}}
\newcommand{\lm}[2]{\lambda_{#1}({#2})}

\newcommand{\tr}{\top}
\newcommand{\DM}{{\cal D}_{\mathbf{M}}}

\newcommand{\smalleqb}[1]
{
	\begingroup
	\makeatletter
	\def
	\f@size{1}
	{#1}
    \endgroup
}

\hypersetup{
	pdfstartview	= {FitH},
	pdftitle			= {Feasibility Based Large Margin Nearest Neighbor Metric Learning},
	pdfsubject		= {ESANN 2018},
	pdfauthor 		= {Babak Hosseini and Barbara Hammer},
	pdfcreator 		= {XeLaTeX},
	pdfproducer		= {LaTeX with hyperref}
}
\urldef{\mailsa}\path|{bhosseini, bhammer}@techfak.uni-bielefeld.de|

\theoremstyle{proposition}

\theoremstyle{theorem}

\theoremstyle{definition}

\theoremstyle{lemma}
\newtheorem{lemma}{Lemma}
\theoremstyle{remark}

\SetKwInput{KwInput}{Input}
\SetKwInput{KwOutput}{Output}
\SetKwInput{KwTask}{Task}
\SetKwInput{KwInit}{Initialization}
\SetKwInput{KwProc}{Loop}
\SetKwInput{KwStepk}{Step K}
\date{\footnotesize{Preprint of the publication~\cite{hosseini2018feasibility}, as provided by the authors.\\
		The final publication is available via \url{https://www.elen.ucl.ac.be/esann/index.php?pg=proceedings} }}
\begin{document}
\title{Feasibility Based Large Margin Nearest Neighbor Metric Learning}

\author{Babak Hosseini$^1$ and Barbara Hammer$^1$
%
\thanks{This research was supported by the Cluster of Excellence Cognitive 
Interaction Technology 'CITEC' (EXC 277) at Bielefeld University, which
is funded by the German Research Foundation (DFG).}
%
\vspace{.3cm}\\
%
CITEC centre of excellence, Bielefeld University\\
Bielefeld, Germany\\
}

\maketitle

\begin{abstract}
Large margin nearest neighbor (LMNN) is a metric learner which optimizes the performance of the popular $k$NN classifier. 
However, its resulting metric relies on pre-selected target neighbors.
In this paper, we address the feasibility of LMNN's optimization constraints regarding these target points, 
and introduce a mathematical measure to evaluate the size of the feasible region of the optimization problem.
We enhance the optimization framework of LMNN by a weighting scheme which prefers data triplets which yield a larger 
feasible region. This increases the chances to obtain a  good metric as the solution of LMNN's problem.
We evaluate the performance of the resulting feasibility-based LMNN algorithm using synthetic and real datasets. The empirical results show an improved accuracy for different types of datasets in comparison to regular LMNN.
\end{abstract}

\section{Introduction}
Metric learning is the idea of finding an efficient metric for a given dataset to provide a more discriminant representation and consequently having a better classification performance. 
In basic terms, it tries to compact points of the same class while increasing the distance between different classes \cite{bellet2013survey}.
A well-known metric learning approach is the Large Margin Nearest Neighbor algorithm (LMNN) \cite{Weinberger2009} which transfers the maximum margin concept of SVM \cite{crammer2001algorithmic} to the $k$-nearest neighbor ($k$NN) framework \cite{cover1967nearest}.
LMNN has been used in many real problems such as face recognition \cite{guillaumin2009you}, motion classification \cite{7344819} and person identification \cite{li2016adaptive}.  
Several improvements have been suggested for the original LMNN approach such as complexity reduction of its optimization \cite{park2011efficiently}, eigenvalue based optimization \cite{ying2012distance}, multi-tasking extension \cite{parameswaran2010large} and hierarchical prepossessing of input data \cite{zhang2017hierarchical}. 
One challenge of LMNN  is  the efficient selection of neighboring targets in its optimization framework \cite{Weinberger2009}. 
As a common strategy, these target points are selected as nearest neighbors from the same class based on the  Euclidean distance. In multiple-pass  LMNN,
the neighborhood is recomputed based on the found distance measure
 to improve the  classification result
 \cite{Weinberger2009,gopfert2016convergence}.

In this paper we focus on the relation between selected neighboring targets and the feasible set of LMNN's optimization problem. We show that wrong choices of targets can severely shrink the regime of feasible solutions of the optimization problem.
%
%
%
%
We introduce a feasibility measure which quantifies the impact  of neighboring points with respect to the
size of the feasible set, and we use this measure as a weighting scheme in a modified version of  LMNN.
\\
\textbf{Road map}: In section \ref{sec:LMNN} we shortly review the original LMNN framework, and afterwards we study the concept of infeasible target neighbors. In section \ref{sec:fbm} we introduce a measure to evaluate the size of target's feasible regions, and
we introduce feasibility-based LMNN in section \ref{sec:fb_lmnn}. We implement our algorithm on synthetic and real datasets in section \ref{sec:exp}, and eventually the conclusion will be made in the last section.
\section{Large Margin Nearest Neighbor Algorithm}\label{sec:LMNN}
Consider the training set $\{(\vec x_i,y_i)\}_{i=1}^n$ with  data vectors $\vec{x}_i \in\mathbb{R}^d$ and their corresponding labels $y_i\in\{1,\ldots,C\}$. LMNN  tries to find a Mahalanobis metric of the form
$
\DM (\vec x_i,\vec x_j)= (\vec x_i-\vec x_j)^\top \mathbf{M}(\vec x_i-\vec x_j)
$
where $\mathbf{M}$ is a positive semidefinite (psd) matrix.
Its objective is to achieve compact neighboring data samples with the same label ({targets})  and far away neighboring points with different labels ({impostors}). 
Define 
$\MC{N}^{k}_{i}$ as the set of points within the $k$-nearest neighbors of $\vec{x}_i$ which have the same class
and  $\MC{I}^{k}_{i}$ as the set of points within the $k$-nearest neighbors of $\vec{x}_i$ which have a different class.
LMNN optimizes the following problem:
\begin{equation}
\begin{array}{ll}
\underset{\mathbf{M}}{\min}  &(1-\mu){\sum}_i\sum_{j \in \MC{N}^{k}_{i}} \DM(\vec x_i,\vec x_j) + \mu \sum_i\sum_{j \in \MC{N}^{k}_{i}}\sum_{l\in\MC{I}^{k}_{i}}\xi_{ijl}\\
\mathrm{s.t.} & \DM(\vec x_i,\vec x_l) - \DM(\vec x_i,\vec x_j) \ge 1-\xi_{ijl}\\
&\xi_{ijl}\ge 0 , ~~{\mathbf{M}} \succeq 0\, ~~ \forall i,j \in \MC{N}^{k}_{i} ,l\in\MC{I}^{k}_{i}
\label{eq:optl}
\end{array}
\end{equation}
where $\mu\in(0,1)$ balances the two objectives.
$\xi_{ijl}$ constitute slack variables of the constraints.
Eq.\ref{eq:optl} constitutes a convex problem with respect to $\mathbf{M}$ if the targets $\MC{N}^{k}_{i}$ and impostors $\MC{I}^{k}_{i}$ are  fixed \cite{Weinberger2009}. Nevertheless, different selections for these initial targets can lead to different solution $\mathbf{M}$. As suggested in \cite{Weinberger2009,gopfert2016convergence} a better strategy is to repeat LMNN's optimization multiple times (multiple-pass LMNN) while updating $\MC{N}^{k}_{i}$ and $\MC{I}^{k}_{i}$ in each run based on the resulting quadratic form $\mathbf{M}$. 
Yet, also this strategy  relies on the quality of the initial selection of these two sets.
%
%
%
%
\section{Infeasible Target Neighbors}\label{sec:inf}
We are interested in the question in which cases feasible solutions of the optimization problem
(\ref{eq:optl}) exist which do not require slack variables $\xi_{ijl}>0$.
This feasible regime is given as
\begin{equation}\label{eq:sm}
S:=\{\MB{M}\in \mathbb{R}^{d\times d}|\MB{M} \succeq 0 , \DM(\vec{x}_i,\vec{x}_j) < \DM(\vec{x}_i,\vec{x}_l)~~ \forall i,j \in \MC{N}^{k}_{i}, l \in \MC{I}^{k}_{i}\}
\end{equation}
For a triplet $i$, $j$, $l$, the metric constraint can be re-written as: 
\begin{equation}\label{eq:trc}
\begin{array}{ll}
&\Tr[\MB{Q}_{ijl}\MB{M}] := \Tr[((\vec{x}_i-\vec{x}_j)(\vec{x}_i-\vec{x}_j)^{\tr}-(\vec{x}_i-\vec{x}_l)(\vec{x}_i-\vec{x}_l)^{\tr})\MB{M}]<0\\ 
\end{array}
\end{equation}
Since  $\MB{M}$ is psd, a psd matrix $\MB{Q}_{ijl}$ leads to the infeasibility of
Eqn.~(\ref{eq:trc}), whereby this fact depends on the triplet $i$, $j$, $l$, only, and not the specific neighborhood. 
In this section, we discuss an extremal case, where the constraint induced by a triplet is infeasible, and 
we propose an according measure which has a clear geometric interpretation in this extremal case.
In the next section, we generalize this measure  to a suitable weighting scheme for more general settings.

A matrix $\MB{Q}:=\MB{Q}_{ijl}$ results from two vectors in the form
$\vec{a}\vec{a}^\tr-\vec{b}\vec{b}^\tr$, i.e.\ its rank is at most $2$.
After matrix transformation if necessary,  we can assume that only the first two dimensions of
the matrix, $\vec{a}$ and $\vec b$ relate to non-zero coefficients.
Denote the two possibly nonzero eigenvalues of $\MB{Q}$ 
as $\lambda_{\MR{min}}(\MB{Q})\le\lambda_{\MR{max}}(\MB{Q})$.
Note that eigenvectors are obviously located in the span of $\vec{a}$ and ${\vec b}$, and 
(after base transformation s.t.\ non-zero coefficients are denoted $(a_1,a_2)$ and $(b_1,b_2)$) 
they have the form
$$\lambda_{\MR{max}/\MR{min}}=
(a_1^2+a_2^2-b_1^2-b_2^2)/2 \pm \sqrt{(a_1^2+a_2^2-b_1^2-b_2^2)^2/4 + (a_1b_2-b_1a_2)^2 }$$
as one can easily infer from the characteristic polynomial of $\MB{Q}$.
Obviously, $\lambda_{\MR{min}}(\MB{Q})\le 0<\lambda_{\MR{max}}(\MB{Q})$
(unless vectors itself are degenerate).
The equality $\lambda_{\MR{min}}(\MB{Q})=0$ corresponds to linearly dependent
vectors $\vec{a}$ and $\vec{b}$, namely the equality $a_1b_2-b_1a_2=0$.
This setting  does not allow a feasible solution without slack variables.
In the following, we will argue that the measure
$r:=-\lambda_{\MR{min}}(\MB{Q})/\lambda_{\MR{max}}(\MB{Q})$ constitutes a reasonable weight vector to measure
the feasibility of the constraint corresponding to $\MB{Q}$ or the size of its feasible domain, respectively.
Obviously, $r=0$ is the case just described, an infeasible setting
due to the geometry of $\vec a=(\vec{x}_i-\vec{x}_j)$ and
$\vec b=(\vec{x}_i-\vec{x}_l)$.

\newpage
\section{Feasibility Measure}\label{sec:fbm}
We start with a general observation:
\begin{lemma}\label{lm:qm}
Denote the eigenvalues of a matrix $\MB{Q}\in \MR{R}^{d \times d}$ by 
$\lambda_1(\MB{Q})\ge
\lambda_2(\MB{Q})\ge\ldots $. The smallest/largest eigenvalue is denoted $\lm{min}{\MB{Q}}$ resp.\
$\lm{max}{\MB{Q}}$.
	For hermitian $\MB{Q} \in \MR{R}^{d \times d}$ and symmetric psd $\MB{M}\in \MR{R}^{d \times d}$,	
it holds
	$\lm{k}{\MB{Q}} \lm{min}{\MB{M}} \le \lm{k}{\MB{QM}}$
	for all $k$.
	\end{lemma}
\begin{proof}
	$\MB{M}$ is psd and $\MB{Q}$ and $\MB{M}$ are symmetric, hence $\lm{k}{\MB{QM}}=\lm{k}{\MB{Q}\sqrt{\MB{M}}\sqrt{\MB{M}}}=\lm{k}{\sqrt{\MB{M}} \MB{Q} \sqrt{\MB{M}}}$ where $\sqrt{\MB{M}}$ is the principal square root of $\MB{M}$. 
	Using the min-max theorem we find
	$
	\lm{k}{\MB{QM}}=
	\underset{{\dim(F)=k}}{\MR{min}} \left( \max_{x\in F\backslash \{0\}} \frac{\langle Q\sqrt{\MB{M}}x,\sqrt{\MB{M}}x\rangle}{\langle\sqrt{\MB{M}}x,\sqrt{\MB{M}}x\rangle}\frac{\langle \MB{M}x,x \rangle}{{\langle x,x \rangle}}\right)
   \geq\lm{min}{\MB{M}}\,\underset{\dim(F)=k}{\MR{min}} \left( \max_{x\in F\backslash \{0\}} \frac{\langle \MB{Q}\sqrt{\MB{M}}x,\sqrt{\MB{M}}x\rangle}{\langle\sqrt{\MB{M}}x,\sqrt{\MB{M}}x\rangle}\right)
	$
    because  $\frac{\langle \MB{M}x,x\rangle}{\langle x,x\rangle} \ge \lambda_{min}(\MB{M})$.\\
Again using the min-max theorem we get $\lm{k}{\MB{QM}} \geq \lm{min}{\MB{M}} \lm{k}{\MB{Q}}$.
\end{proof}
Based on Lemma~\ref{lm:qm} we have $\lm{max}{\MB{Q}} \lm{min}{\MB{M}} \leq \lm{max}{\MB{QM}}$
for $\MB{Q}:=\MB{Q}_{ijl}$ as specified above.
In the setting $\lm{min}{\MB{Q}}<0< \lm{max}{\MB{Q}}$, we can use
\cite{zhang2006eigenvalue}(corollary 10) to infer $\lm{min}{\MB{Q}} \lm{max}{\MB{M}} \leq \lm{min}{\MB{QM}}$.
Combining these two inequalities results in the inequality
\begin{equation}
\lm{min}{\MB{Q}} \lm{max}{\MB{M}}+\lm{max}{\MB{Q}} \lm{min}{\MB{M}} \leq \Tr(\MB{QM})
\label{sic}
\end{equation}
Eq.~\ref{eq:trc} induces the objective $\Tr(\MB{QM})<0$, hence the left hand side of Eq.~\ref{sic} should be negative, i.e.\  $-\frac{\lm{min}{\MB{Q}}}{\lm{max}{\MB{Q}}} > \frac{ \lm{min}{\MB{M}}}{\lm{max}{\MB{M}}}$.
Hence a triplet $i,j,l$ with a small value  $r=-\frac{\lm{min}{\MB{Q}}}{\lm{max}{\MB{Q}}}$ imposes a tight constraint on the eigenvalue formation of $\MB{M}$, hence we expect an induced small feasible set $S_{ijl}$. 
Note that the feasible domain $S$ results as intersection of the feasible sets $S_{ijl}$.
We  include this observation and the according measure $r=r_{ijl}$ into the optimization framework in the form of an 
according weighting.
\section{Feasibility Based LMNN} \label{sec:fb_lmnn}
For a vector $\vec{x}_i$ and a given target $\vec{x}_j \in \MC{N}^{k}_{i}$ we define $R_{ij}:=\underset{\vec{x}_l \in \MC{I}^{k}_{i}}{\min} (r_{ijl})$.
We formulate \emph{feasibility-based LMNN}\/ as the following LMNN problem which incorporates according weights of the objective:
\begin{equation}\label{eq:fblmnn}
\begin{array}{ll}
\underset{\mathbf{M}}{\min}  &(1-\mu){\sum}_i\sum_{j \in \MC{N}^{k}_{i}} R_{ij}\DM(\vec x_i,\vec x_j) + \mu \sum_i\sum_{j \in \MC{N}^{k}_{i}}R_{ij}\sum_{l\in\MC{I}^{k}_{i}}\xi_{ijl}\\
\mathrm{s.t.} & \DM(\vec x_i,\vec x_l) - \DM(\vec x_i,\vec x_j) \ge 1-\xi_{ijl}\\
&\xi_{ijl}\ge 0 , ~~{\mathbf{M}} \succeq 0, ~~ \forall i,j \in \MC{N}^{k}_{i} ,l\in\MC{I}^{k}_{i}
\end{array}
\end{equation}
Unlike original LMNN, infeasible or challenging triplets carry less weighting in this formulation.
We dub the resulting method
 FB-LMNN.
 FB-LMNN is implemented by first determining the neighborhood, computing corresponding weights $R_{ij}$, and then 
 solving the convex optimization problem w.r.t.\ matrix $\mathbf{M}$. In addition, a multiple passes strategy
 can
 be used to increase the resulting accuracy \cite{Weinberger2009}.

\section{Experiments}\label{sec:exp}
We evaluate our algorithm on both synthetic and real data, and compare it with $k$NN, single-path LMNN (SP-LMNN), multiple passes LMNN (MP-LMNN)\cite{Weinberger2009} and multi-class SVM \cite{crammer2001algorithmic}.
For LMNN  we use a neighborhood size  $k=5$ and weighting $\mu=0.5$. Evaluation is done in a 10-fold cross validation.
SVM uses the respective best result obtained with a linear, RBF-, or polynomial kernel.

\subsection{Synthetic Data}
The synthetic dataset is a variation  of the 2D zebra stripe data from \cite{Weinberger2009} in which two classes of data alternate  (Fig.\ref{figz1}-left). 
In contrast to \cite{Weinberger2009}, the nearest targets to each data point \xj{i} are located on the same stripe.
On each stripe, the impostors and targets are almost placed on a straight line, resulting in very tight or infeasible constraints of the optimization framework. 
Even multiple-pass LMNN (Fig.\ref{figz1}-middle) does hardly change this selection of impostors and targets.  Consequently multiple-pass LMNN converges to a non-optimal solution $\MB{M}$ with classification accuracy of $23.51\%$ (almost the same as $k$NN's).

On the other hand, FB-LMNN assigns small $R_{ij}$ weights to pairs within the same stripe  while bigger weights to pairs in neighbored stripes. Therefore it obtains a different matrix $\MB{M}$ which results in a different scaling of the space (Fig.\ref{figz1}-right), and a classification accuracy of $72.21\%$. 
\begin{figure}[tb]
	\minipage{0.45\textwidth}
	\includegraphics[width=\linewidth]{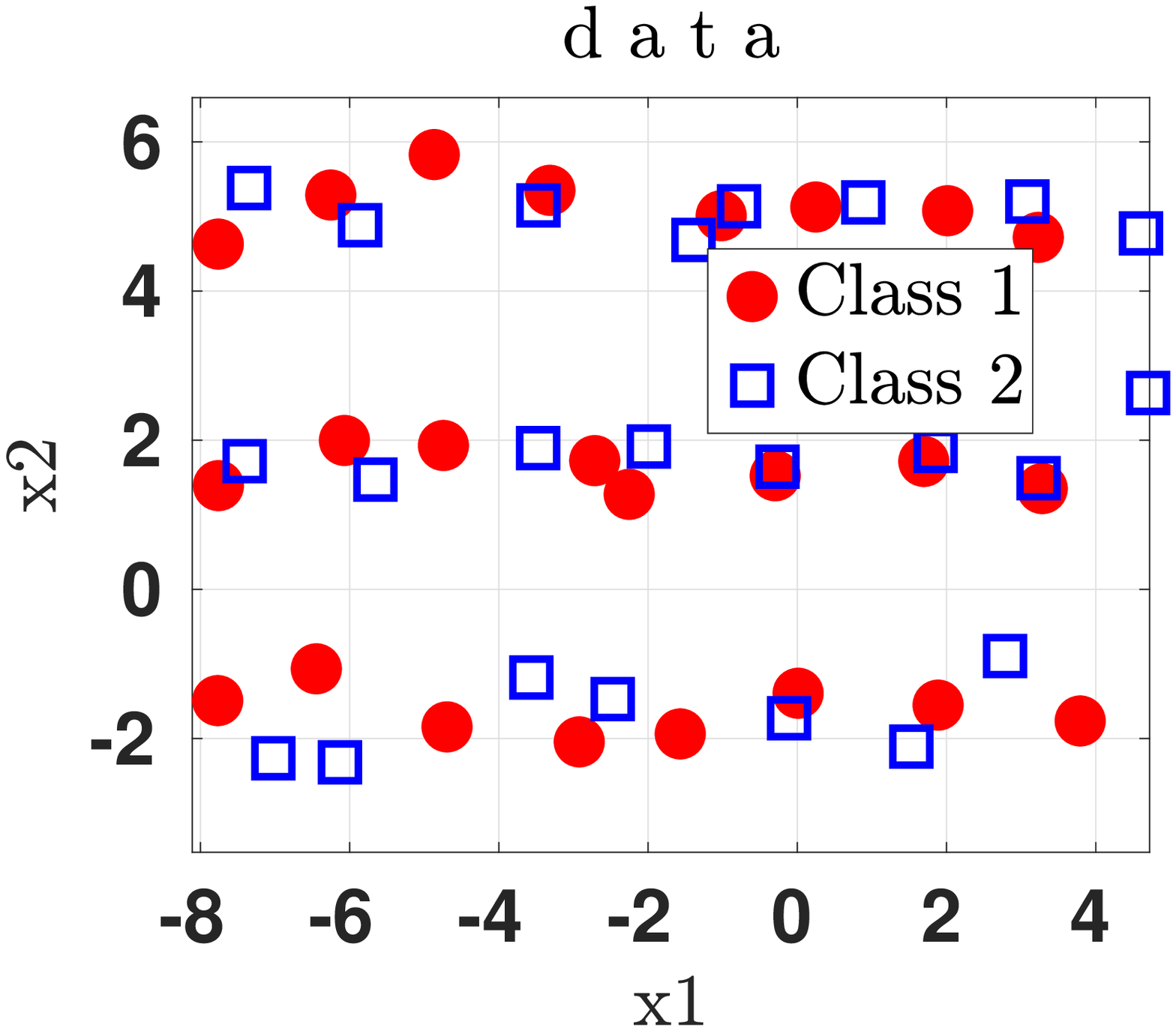}
	\endminipage
	\hspace{1mm}
	\minipage{0.45\textwidth}
	\includegraphics[width=\linewidth]{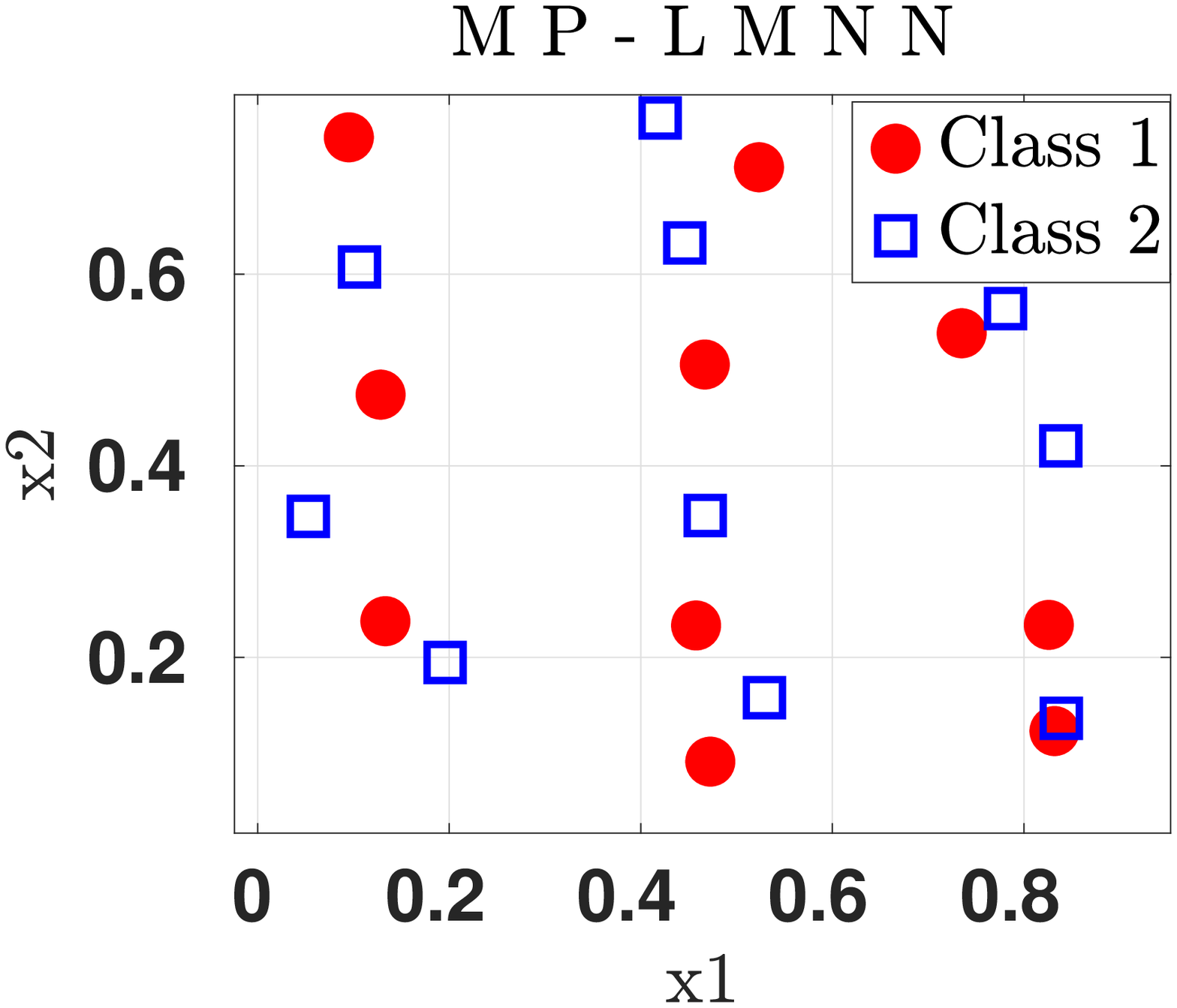}
	\endminipage
    \vspace{0.5cm}
	\minipage{0.45\textwidth}%
	\includegraphics[width=\linewidth]{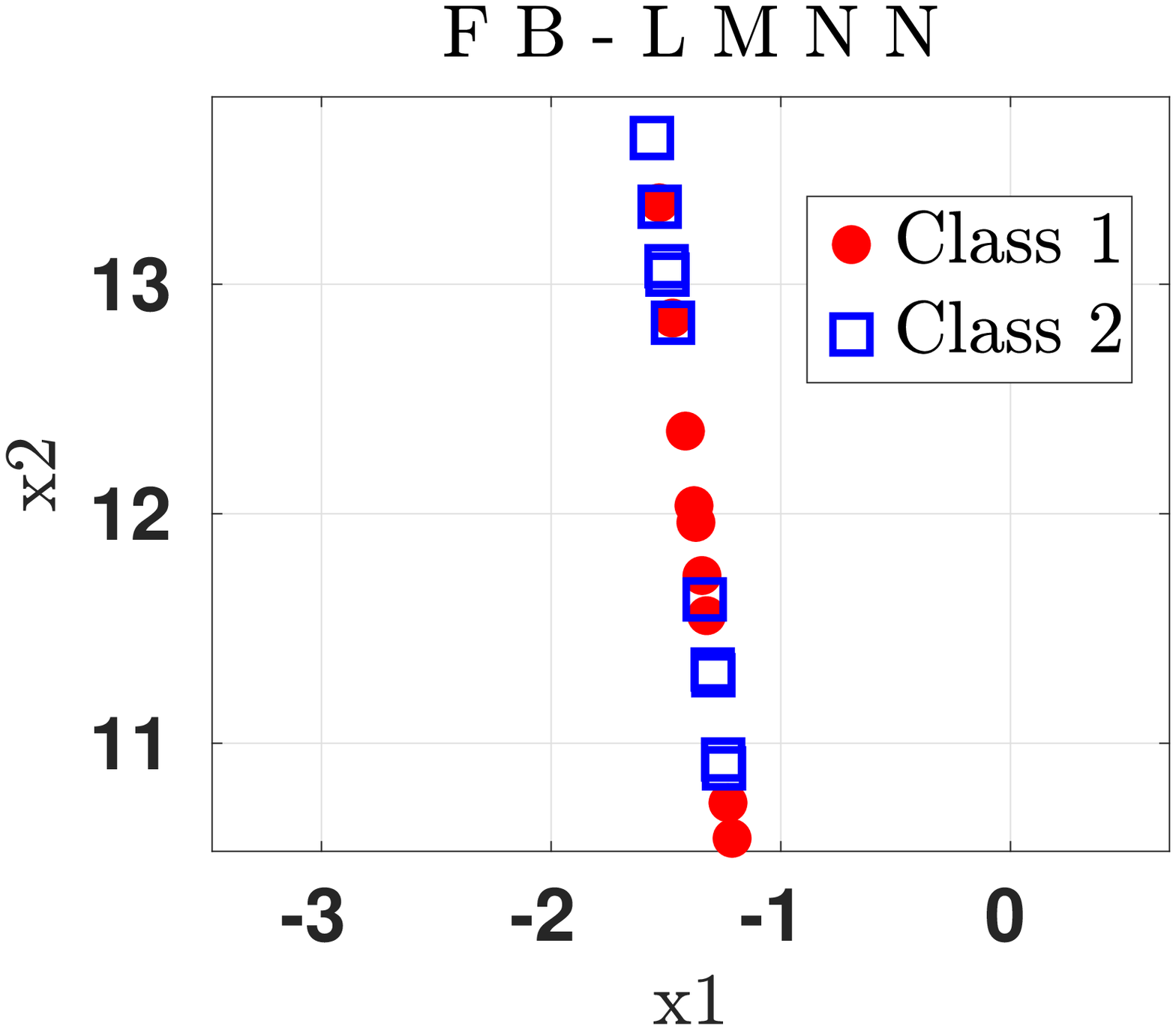}
	\endminipage
		\caption{Zebra dataset (top-left), MP-LMNN (top-right) and FB-LMNN (bottom)}
		\label{figz1}
\end{figure}
\subsection{Real Datasets}
Real datasets are mostly taken from the UCI repository library\footnote{https://archive.ics.uci.edu/ml/datasets.html}; in addition, we consider the
extended Yale face dataset\footnote{http://vision.ucsd.edu/~iskwak/ExtYaleDatabase/ExtYaleB.html} 
and the MNIST handwritten digits\footnote{http://yann.lecun.com/exdb/mnist/}, which constitute benchmarks in this domain.
 The selected datasets cover different application areas and types of the data set.
 For the Yale face, MNIST, and isolet datasets we follow the procedure proposed in \cite{Weinberger2009} as 
 regards preprocessing by means of dimensionality reduction and the cross-validation procedure for evaluation.

Results are shown in Table \ref{tab:res}. FB-LMNN significantly surpasses MP-LMNN in some cases, demonstrating the
effectivity of the proposed feasibility measure $R$. 
For a few data sets, such as the Yale dataset, no significant difference is observed.
Interestingly, results as obtained by SVM are mostly on par, but in some cases worse as compared to FB-LMNN,
whereby SVM restricts to the standard Euclidean metric. It would be an interesting endeavor to test the effect of the 
metric learned by FB--LMNN on the result of SVM.
\begin{table}[tb]
	\renewcommand{\arraystretch}{0.9}
	\small
	\caption{Classification accuracy(\%) and datasets' characteristics.}
	\label{tab:res}
	\centering	
	\begin{tabular}{|c|c|c|c|c|c|c|c|} 
			\hline
			
				 dataset & class &  dimension & $k$NN & sp-lmnn & mp-lmnn & \textbf{fb-lmnn} & SVM\\
		\hline		
		Zebra  & 2 & 2& 21.31 & 22.41 &23.51 & \textbf{72.21} & 50.82\\
		\hline		
		Wine & 3&13 &  76.20 & 92.84 &93.91 & \textbf{98.77} & 78.23\\
		\hline		
		Balance &3& 4&  83.42 & 88.45 &94.03 & 96.08 & \textbf{97.5}\\
		\hline		
		B.~Cancer  &20 & 30& 94.66 &94.88 & 96.68 & \textbf{97.07} & 78.49\\
			\hline
			Car Eval. & 4& 6& 92.57 & 95.12 & \textbf{98.32} & \textbf{98.4} & 60.08\\
			\hline
	Tic-Tac-Toe &2 & 9&  87.42 & 91.46 &97.66 & \textbf{98.13} & 85\\
			\hline
			Hepatitis & 2&17 & 84.16 & 84.46 &84.46 & \textbf{90} & 79.11\\
			\hline
			iris & 3 & 4&  94.93 & 95.02 &95.61 & \textbf{96.05} & \textbf{96.13}\\
			\hline
			isolet & 26 & 172 &  91.02 & 95.64 & 95.70 & \textbf{96.85} & 96.60\\
			\hline
			YFace & 38 & 300 &  89.21 &  94.10 & \textbf{94.48} & \textbf{94.48} &  84.78 \\
			\hline
			MNIST &10 &164 &  97.57 &  98.28 & 98.31 & \textbf{98.92} & 98.80 \\
			\hline
	\end{tabular}
\end{table}
\section{Conclusion}\label{sec:conc}
In this paper we studied the role of  target neighbors $\MC{N}^{k}_{i}$ on the feasibility of the  constraints in LMNN's optimization problem.
We proposed a quantitative  measure for the degree of  feasibility of triplets of a target and impostor for a data point, and we
demonstrated that this measure constitutes an efficient and effective weighting scheme to be integrated into LMNN's optimization. The results of several experiments clearly demonstrate the effect of the proposed technology.

\newpage
\bibliographystyle{unsrt}
\bibliography{Ref4Papers_CS}


\end{document}